%% file: paper.tex
\documentclass[conference]{IEEEtran}
\usepackage {graphicx} % For � kunne inkludere grafikk
\usepackage {amsthm}
\usepackage[fleqn] {amsmath}
\usepackage {amsfonts, amssymb}
\usepackage{multirow}

\usepackage{enumerate}
\usepackage{color}
\usepackage{algorithmic} 
\usepackage[]{algorithm}
\newtheorem{theorem}{Theorem}

\newtheorem{lemma}[theorem]{Lemma}
\newtheorem{proposition}[theorem]{Proposition}
\newtheorem{definition}{Definition}

\title{Secondary Access to Spectrum with SINR Requirements Through Constraint Transformation}
\author{\IEEEauthorblockN{Brage Ellings\ae ter\IEEEauthorrefmark{1}\IEEEauthorrefmark{2}}
\IEEEauthorblockA{\IEEEauthorrefmark{1}UNIK-University Graduate Center at Kjeller, University of Oslo, Norway}
\IEEEauthorblockA{\IEEEauthorrefmark{2}Norwegian Defence Research Establishment}\texttt{brage@unik.no}}%Email: \href{mailto:brage@unik.no}{brage@unik.no}}
\begin{document}
\maketitle
\begin{abstract}
In this paper we investigate the problem of allocating spectrum among radio nodes under SINR requirements. This problem is of special interest in dynamic spectrum access networks where topology and spectral resources differ with time and location. The problem is to determine the number of radio nodes that can transmit simultaneously while still achieving their SINR requirements and then decide which channels these nodes should transmit on. Previous work have shown how this can be done for a large spectrum pool where nodes allocate multiple channels from that pool which renders a linear programming approach feasible when the pool is large enough. In this paper we extend their work by considering arbitrary individual pool sizes and allow nodes to only transmit on one channel. Due to the accumulative nature of interference this problem is a non-convex integer problem which is NP-hard. However, we introduce a constraint transformation that transforms the problem to a binary quadratic constraint problem. Although this problem is still NP-hard, well known heuristic algorithms for solving this problem are known in the literature. We implement a heuristic algorithm based on Lagrange relaxation which bounds the solution value of the heuristic to the optimal value of the constraint transformed problem. Simulation results show that this approach provides solutions within an average gap of 10\% of solutions obtained by a genetic algorithm for the original non-convex integer problem.
\end{abstract}

%\begin{IEEEkeywords}
%Spectrum Allocation, Cognitive Radio, Dynamic Spectrum Access, Linear Relaxation
%\end{IEEEkeywords}

\section{Introduction}
An important aspect of cognitive radio and dynamic spectrum access is how to distribute resources among users. Resource allocation in such systems differ from conventional wireless systems because the resources available changes with time and location, yielding an optimal allocation hard. In dynamic spectrum access scenarios, management of spectrum is usually divided between two approaches \cite{Lehr}: (i) property rights/licensed/exclusive use and (ii) open access/unlicensed/spectrum commons. In the latter all devices are treated equal and sharing of the spectrum is done through protocols and etiquettes. Although these protocols and etiquettes may be designed to achieve some overall system performance, quality of service can not be guaranteed in this scenario as everything will be best effort.

In the licensed scenario, devices can buy and trade access to spectrum. In such a scenario there exists central entities which grant access to spectrum on a legal basis and central control. This approach has led to numerous publications on spectrum trading, including demand \cite{Niyato2008}, pricing \cite{Ileri} and auctions of spectrum \cite{Gandhi2007}. In such a scenario, quality of service can be guaranteed, and in fact must be guaranteed in order for the market to function. A natural problem for spectrum holders is then how to maximize revenue.  

In this paper we assume secondary users can buy access to spectrum for a given time period. When requesting spectrum, the secondary users issue a statement about the necessary quality-of-service for their operation. In this paper we limit our focus to the case where such a QoS metric consists of a signal-to-interference plus noise ratio (SINR). Each user is willing to pay a certain amount for its service. The goal of the spectrum holder is then to maximize the number of users that can be granted access to spectrum, while still guaranteeing the QoS for these users. This becomes an admission control and spectrum allocation problem, which is difficult to solve.

%In fact most resource optimization problems in such dynamic networks have been shown to be NP-hard, e.g. the sum-rate problem which is NP-hard for both power and frequency allocation \cite{wei}.

The difficulty of such a problem lies in the nature of the wireless channel: interference. The accumulative nature of interference causes the coupling between decision variables to be non-linear and usually the optimization problem is non-convex (and NP-hard). A central approach is usually assumed to solve such problems close to optimality, and techniques such as genetic algorithms or geometric programming can be used \cite{4275017}. Another approach is to solve the problem by a distributive approach by investigating the Karush-Kuhn-Tucker (KKT) necessary conditions for optimality \cite{wang}. However, finding the KKT multipliers is complex, and in addition the objective function must be continuously differentiable and the problem has to be convex for the KKT conditions to also be sufficient conditions for optimality \cite{Boyd2004}.

Based on a recent work \cite{optimus}, we use a constraint transformation to decouple the joint admission control and spectrum allocation problem into two disjoint problems. We show that the admission control problem must be solved by a central entity in order to guarantee any QoS to the secondary users, while the spectrum allocation problem can be solved in a distributed manner close to, and in some cases to, optimality. This has the advantage of decreasing complexity of the central controller and the option of delegating part of the processing to the users while still having a central control.

The rest of this paper is organized as follows: Section \ref{sec:rel-work} provides an overview of related work, Section \ref{sec:pre} defines the system model and problem statement, Section \ref{sec:prob} presents the constraint transformation used in access control, Section \ref{sec:mkp} describes the the heuristic algorithm presented in this paper, Section \ref{sec:channel-selection} describes the channel selection algorithm, Section \ref{sec:performance} provides performance evaluations and Section \ref{sec:conc} concludes the paper.
%In this paper we investigate the spectrum allocation problem (SAP) under signal to noise plus interference (SINR) constraints. We assume $N$ transmitter and receiver pairs (users) are randomly located in an area, each with a specific SINR requirement. The problem is to find the maximum number of users that can transmit simultaneously while achieving their SINR requirements and assign spectrum to these users. In many scenarios, such as voice or video streaming, this is more relevant than for instance the maximum sum-rate problem or the proportional fairness problem, as it does not really matter if a user can achieve 10 b/s/Hz if it only needs 5 b/s/Hz and a user does not care whether another user achieves 2 times its rate as long as it achieves its own requirement.

%One way to combine the two would be to maximize sum-rate with individual SINR constraints for each user, as for instance done in \cite{mapel} for the power allocation problem. However, this approach leads to a non-convex objective function with non-convex constraints, instead of a linear objective function with non-convex constraints and therefore renders a linear approximation of the constraints inefficient.

\section{Related Work}
\label{sec:rel-work}
Due to the reuse factor used in early mobile phone systems and in many of today's broadcast systems, the spectrum allocation problem (SAP) has been extensively studied in the literature, see for instance \cite{hale} and references therein. However, the SAP problem described for cellular networks with reuse factors differ from ours by the fact that these problems dealt with fixed topology and fixed resource allocations. E.g. each base station was allocated $x$ amount of spectrum, the problem was finding this $x$ for each base station. Also, for cellular networks with fixed topology, solution time is not of the essence as the topology does not change. In more dynamic networks, such as dynamic spectrum access networks, a popular approach to simplifying the problem is to simplify the interference model. One frequently used approach is to approximate the interference model by a graph with pairwise constraints. E.g. node $a$ cannot transmit on the same frequency as node $b$. This is done in for instance \cite{Subramanian} and \cite{4658258}. However, the inefficiency of such graph-based models has also been analyzed quite extensively \cite{Gronkvist:2001:CGI:501416.501453}\cite{Moscibroda:2006p137}.

To model interference and successful reception we use the physical SINR model \cite{4146676}, which accounts for accumulated interference as opposed to graph-based models. In recent years this model has been extensively studied for link scheduling in wireless networks, where approximation guarantees on the optimal solution have been shown to exist \cite{Goussevskaia:2007:CGS:1288107.1288122}\cite{Moscibroda:2007:WCW:1236360.1236362}\cite{Halldorsson:2011:WCO:2133036.2133155}. The best known solutions yield a constant factor $\mathcal{O}(1)$ approximation on optimality. We extend these works by considering multiple channels and individual SINR requirements.

\cite{optimus} considers the case of distributing spectrum among users with cumulative interference and a universal SINR constraint for all users. They assume a common set of channels $K$ available to each user where $|K| \sim 1000$ and the goal is to allocate a fraction of the total number of channels to each user according to some objective function. To overcome the cumulative interference problem they propose a novel linear relaxation of the interference constraints such that the problem becomes a linear program. The linear relaxation of the interference constraints introduced in \cite{optimus} forms the basis of our work. Compared to \cite{optimus}, we extend their work by considering individual channel sets and SINR requirements for each user, and also assume the individual channels sets can be arbitrarily small. Instead of allocating a fraction of a common set of channels to each user, we assume a user can obtain a maximum of one channel. These assumptions changes the problem from a linear program in \cite{optimus} to a binary quadratic constraint (BQC) problem. %In addition we show how the SAP with SINR constraints can be solved in a distributed manner. The motivation for these extensions is that it is more compliant with most of today's standardization activities on dynamic spectrum access networks, such as the IEEE 802.22 standard \cite{wran,Srinivasan}.

In \cite{Xiang2010} spectrum allocation among femtocells is considered, where the goal is to maximize overall system rate while guaranteeing a minimum SINR at each femtocell user. This may seem as a more generalized approach then the one considered in this paper. However, such a formulation has certain drawbacks: (i) with such a formulation each femtocell user \textit{must} transmit, otherwise this problem is equivalent to the general sum-rate problem which is known to be unfair (i.e. most of the resources are allocated to users with high channel gain). (ii) It is known that in many cases some users must stay silent for a feasible solution to exist. Thus the requirement that each user must transmit on some channel limits the cases where a feasible solution exist. However, from a spectrum holders perspective, if there is demand then a feasible solution exists as one could simply allocate all spectrum to the highest bidder.

The work most closely related to ours is \cite{Hoang2008}. In \cite{Hoang2008} the goal is to allocate spectrum to as many cognitive users as possible while achieving a SINR requirement at each user and satisfying some interference limit at the primary users. It is shown that the general problem formulation is NP-hard and thus a greedy heuristic is proposed. As the problem formulation is still NP-hard in this paper, we also rely on a heuristic algorithm to solve the optimization problem efficiently. However, we consider a case where primary users are only protected through a maximum power constraint at each user, such as in the TV white spaces. By this reduction we can use a constraint transformation such that we are able to utilize heuristics with a much lower complexity bound and by comparing sufficient and necessary conditions we are able to upper bound the optimal solution of our problem in certain cases.

\section{System Model and Problem Formulation}
\label{sec:pre}
\subsection{System Model}
A set of $N$ transmitters wants to allocate spectrum to support transmission to their receivers. A pair consisting of a transmitter and a desired receiver is denoted as a user. These users can be one-hop links in ad-hoc networks, or they can be access points to terminal links. In latter case, if different terminals connected to the same access point uses different frequencies, a user has to be defined for each terminal to support this scenario. However, in general this does not affect the spectrum allocation framework, as this only depends on how interference between these users are defined.

To acquire access to spectrum, the users contacts a spectrum holder through some means of communication. The accepted approach to grant access to secondary spectrum is through a database approach as currently done by the FCC \cite{FCC4}. We assume the system consists of a set of channels $\mathcal{K}$. By providing location information to a spectrum holder database, each user has potential access to a subset $\mathcal{K}_i$ of channels along with a maximum power constraint $P_i$ which guarantees that primary users are not affected.
Thus we have that $\mathcal{K}_i \subseteq \mathcal{K}$ and $K_i = |\mathcal{K}_i|$. Each user can select one channel among the $K_i$ channels, where each channel has the same bandwidth and propagation characteristics. The SINR of user $i$ on channel $k$ is given as
\begin{equation}
SINR_{i}^k = \frac{g_{i,i}P_i}{\sigma^2 + \sum_{j\neq i} a_j^kg_{j,i}P_j}
\end{equation}
where $a_j^k \in \{0,1\}$ is 1 if user $j$ transmits on channel $k$, $g_{j,i}$ is the channel gain between node $j$ and node $i$ and $\sigma^2$ is the noise variance.

Let $S_i = g_{i,i}P_i$ and $I_{j,i} = g_{j,i}P_j$. The SINR$_i^k$ is then given as
\begin{equation}
SINR_{i}^k = \frac{S_i}{\sigma^2 + \sum_{j\neq i} a_j^kI_{j,i}}
\end{equation}
We assume each has a SINR requirement $\beta_i$.%The objective is to allocate channels among the users to maximize the number of users that achieve their SINR requirements $\beta_i$.

\begin{definition}
Let $\mathbb{A}$ define a spectrum allocation,
\begin{equation}
\mathbb{A} = \{a_i^k\}, 1\leq i\leq N,1\leq k \leq K_i
\end{equation}
where $a_i^k=1$ indicates that user $i$ selects channel $k$, $a_i^k = 0$ otherwise.
\end{definition}

\begin{definition}
With the SINR requirements of each user, a spectrum allocation $\mathbb{A}$ is \textbf{successful} if $a_i^k=1$ implies $SINR_i^k\geq \beta_i$.
\end{definition}

\subsection{Problem Formulation}
As mentioned above each user has a SINR requirement. By contacting spectrum holder database, each user is willing to pay a fee for access to the spectrum, given that its SINR requirement is met. To provide a price function versus time and performance is out-of-scope of this paper, we only assume it is a non-decreasing function of SINR. The goal of the spectrum holder is thus to maximize its revenue by allowing an optimal set of users access to the spectrum under their individual SINR constraints.

Since a user can only use one frequency, we introduce a variable $x_i \in [0,1]$ where
\begin{equation}
x_i = \sum_k a_i^k 
\end{equation}
The problem can then be formally defined as follows
\begin{align}
\max& \sum_i r_ix_i \label{eq:org-prob}\\
\text{s.t.:}& \nonumber \\
&a_i^k(SINR_i^k -\beta_i)\geq 0 \label{eq:constraint-non-linear} \\
&x_i = \sum_k a_i^k \\
&0\leq x_i\leq 1, a_i^k\in\{0,1\} \label{eq:a-0-1}
\end{align}
where $r_i$ is the revenue of the spectrum holder given that user $i$ is granted access to the spectrum. This problem is a non-convex 0,1 problem, as (\ref{eq:constraint-non-linear}) is a non-convex constraint and the values of $a_i^k$ is either 0 or 1. The difficulty of the problem lies in the large unknown non-convex solution space due to the large number of possible variable configurations. For instance if there are 10 users with 5 available channels each, there are $5^{10}$ possible combinations. 

As noted in numerous papers, this problem belong to the class of NP-hard problems. Therefore the time required to solve the problem to optimality increases exponentially with the number of constraints. The approached used in this paper consists of a decoupling of the problem into an access control problem and channel allocation problem. The reason for this approach is twofold: first, through a constraint transformation a sufficient condition for a feasible solution can be stated which significantly reduces the solution space and for which efficient heuristics exist. Secondly, given a set of users which is known to contain a feasible solution, a distributed approach to channel selection can be done which in general perform close to optimal and in certain cases is optimal.

Thus the channel selection part of the problem can be done both centrally and distributively. This has the advantage of being able to offload some of the processing requirements at the central entity to the users in the field. On the other hand, access control cannot be done in a distributed manner while still guaranteeing QoS at the users as "global" information about the users is required. In Section \ref{sec:performance}, we show how limiting knowledge to consist of only certain neighboring users affects performance as would be the case in a distributed implementation. 

\section{Access Control}
\label{sec:prob}
\subsection{Constraint Transformation with Equal Channel Sets}
To overcome the large solution space in the original problem (\ref{eq:org-prob})-(\ref{eq:a-0-1}), we transform constraint (\ref{eq:constraint-non-linear}) to a binary quadratic constraint, which although is not linear, reduces the possible combinations of variables to $2^{N}$ for which an efficient heuristic algorithm can find good solutions. The constraint transformation was introduced in \cite{optimus}, but we have modified it to our problem. Instead of transforming (\ref{eq:constraint-non-linear}) from a non-convex constraint to a linear constraint, which is possible in \cite{optimus} due to the fractional channel set allocation which renders a linear approximation possible, we transform the constraint to a binary quadratic constraint due to our single-tone approach. 

%The main difference is that instead of transforming (\ref{eq:constraint-non-linear}) from a non-convex constraint to a linear constraint, we transform the constraint to a binary quadratic constraint. Also our transform is applicable to the single-tone approach done in this paper, instead of the fractional set approach done in \cite{optimus}.

The constraint given in (\ref{eq:constraint-non-linear}) can be approximated by
\begin{equation}
x_i(1 + \sum_{j\neq i} x_j\frac{I^{+}_{j,i}}{I_i^{\max}}) \leq K_i
\label{eq:linearized-c-t}
\end{equation}
where $I_i^{\max} = \frac{S_i}{\beta_i}-\sigma^2$ and $I^{+}_{j,i} = \min(I_i^{\max},I_{j,i})$. The problem then becomes
\begin{align}
&\max \sum_i r_ix_i \label{eq:max-z-t}\\
&\text{s.t.} \nonumber \\
&x_i(1 + \sum_{j\neq i} x_j\frac{I^{+}_{j,i}}{I_i^{\max}}) \leq K_i, x_i\in\{0,1\}  i=1,...,N \label{eq:x-0-1-t}
%&x_i\in \{0,1\} \label{eq:x-0-1}
\end{align}
If $x_i$ equals 1, this means user $i$ can achieve its SINR by selecting some channel. It does not, however, find the channel that should be selected. Let $z'$ be the solution to the optimization problem given in (\ref{eq:max-z-t})-(\ref{eq:x-0-1-t}), i.e. the number of users that can obtain their SINR requirement. We then have the following proposition:
\begin{proposition}
$z'$ is a feasible solution to the original optimization problem (\ref{eq:org-prob})-(\ref{eq:a-0-1}).
\label{prop:1}
\end{proposition}
\begin{IEEEproof}
See Appendix \ref{app:proof-prop1}
\end{IEEEproof}

\subsection{Constraint Transformation With Unequal Channel Sets}
Although the solution to (\ref{eq:max-z-t})-(\ref{eq:x-0-1-t}) finds a feasible solution to the original problem, we conjecture that the constraint given in (\ref{eq:linearized-c-t}) is too restrictive in some cases. This is especially the case when the different users have different subsets of channels available to them, which is the case in many secondary scenarios as spectral resources depends can depend on location and device specifications. In the proof of Proposition \ref{prop:1}, a feasible solution is shown to exist by bounding the number of blocked channels, where all users are assumed to be able to block all channels at any given user if they are allowed to transmit. However, if two users don't share any common channels they cannot block each other. Also, if two users only share a subset of channels it is less likely that they will contribute the blocking of channels. Thus we propose a new constraint for this scenario which takes into account the probability of two users blocking each other based on their respective channel sets which is yields a feasible solution in the asymptotic regime almost surely. The new constraint is given by
\begin{equation}
x_i(1 + \sum_{j\neq i} x_j\frac{I^{+}_{j,i}|\mathcal{K}_j\cap \mathcal{K}_i|}{I_i^{\max}K_j}) \leq K_i
\label{eq:linearized-c}
\end{equation}
and thus the new optimization problem becomes
\begin{align}
&\max \sum_i r_ix_i \label{eq:max-z}\\
&\text{s.t.} \nonumber \\
&x_i(1 + \sum_{j\neq i} x_j\frac{I^{+}_{j,i}|\mathcal{K}_j\cap \mathcal{K}_i|}{I_i^{\max}K_j}) \leq K_i, i=1,...,N \\
&x_i\in\{0,1\},  i=1,...,N \label{eq:x-0-1}
%&x_i\in \{0,1\} \label{eq:x-0-1}
\end{align}
\begin{proposition}
Let $z^{*}$ be the solution to the problem given in (\ref{eq:max-z})-(\ref{eq:x-0-1}). $z^{*}$ is a feasible solution to the original optimization problem (\ref{eq:org-prob})-(\ref{eq:a-0-1}) almost surely when $N\rightarrow \infty$.
\label{prop:a}
\end{proposition}
\begin{IEEEproof}
See Appendix \ref{app:proof-propA}
\end{IEEEproof}

%\begin{corollary}
%Let $z$ be an optimal solution to the original optimization problem. Then $E[z]\geq z^{*}$.
%\end{corollary}
We now have two new 0,1 non-linear problems to compute the maximum revenue that a spectrum leaser can achieve while satisfying the requirements of a successful alloation. Specifically, the problems are modified versions of the multidimensional 0,1 knapsack problem (MKP), with the difference that instead of linear constraints we have binary quadratic constraints. As MKPs are NP-hard \cite{kellerer}, so are our problems. Thus solving the above problems is not easy. By exploiting the structure of our particular problems we use a Lagrange relaxation inspired by the one introduced in \cite{magazine1984319} for the standard MKP. A heuristic based on Lagrange relaxation has desirable properties such as low complexity (running time is $\mathcal{O}(2N^2)$ for our problem) and an upper bound for $z^{*}$ without resorting to an LP relaxation of the problem.

 %along with a heuristic algorithm introduced in \cite{volgenant}. Although this heuristic algorithm is not the most efficient in terms of finding good solutions in the literature (see for instance  and references therein), it have some desirable properties such as low complexity (running time is $O(2N^2)$ for our problem) and the possibility of a distributed implementation. In addition we can also find an upper bound for $z^{*}$ without resorting to an LP relaxation of the problem.

\section{A Heuristic using Lagrange Relaxation}
\label{sec:mkp}
To simplify notation we set $a_{ij} = \frac{I^{+}_{j,i}|\mathcal{K}_j\cap \mathcal{K}_i|}{I_i^{\max}K_j}$ (or $a_{ij} = \frac{I^{+}_{j,i}}{I_i^{\max}}$ if solving (\ref{eq:max-z-t})-(\ref{eq:x-0-1-t})). A Lagrange relaxation of our problem is
\begin{align}
\max_{\mathbf{x}}&\bigl\{\sum_{i=1}^N r_ix_i + \sum_{i=1}^N \lambda_i (K_i-x_i(1+\sum_{j\neq i}^N a_{ij}x_j))\bigr\}\label{eq:rl-prob} \\
\text{s.t.:}&\nonumber \\
&\mathbf{x}\in\{0,1\}^N \text{ and } \mathbf{\lambda}\geq 0. \label{eq:rl-x}
\end{align}
Let $\mathbf{x'}$ be a solution to (\ref{eq:max-z})-(\ref{eq:x-0-1}) (or (\ref{eq:max-z-t})-(\ref{eq:x-0-1-t})) with $z(\mathbf{x'})$ as its value. Let RL$(\mathbf{x'})$ be the value of the Lagrange relaxation problem. Clearly RL$(\mathbf{x'})\geq z(\mathbf{x'})$, since when $\mathbf{x'}$ satisfies (\ref{eq:linearized-c}), $K_i-x_i(1+\sum_{j\neq i}^n a_{ij}x_j)\geq 0$. Therefore the optimal value of (\ref{eq:rl-prob})-(\ref{eq:rl-x}) is larger than than $z^{*}$ and it is therefore a relaxation.

It is easy to see that this maximization problem is equivalent to the following
\begin{align}
\max_{\mathbf{x}}&\bigl\{\sum_{i=1}^N \bigl(r_i-\lambda_i(1+\sum_{j\neq i}^N a_{ij}x_j)\bigr)x_i\bigr\} \label{eq:mod-langrange}\\
\text{s.t.:}&\nonumber \\
&\mathbf{x}\in\{0,1\}^N \label{eq:x-0-1-2}
\end{align}
since $\lambda_iK_i$ are constants. It is also easy to see that (\ref{eq:mod-langrange})-(\ref{eq:x-0-1-2}) has the trivial solution
\begin{equation}
x^{*}_i = \left\lbrace\begin{array}{c c}{1} & {\text{if }(r_i-\lambda_i(1+\sum_{j\neq i}^N a_{ij}x_j)\bigr)>0} \\ {0} & {\text{otherwise}}\end{array}\right.
\label{eq:determine-x}
\end{equation}
The difficulty lies in finding values for the Lagrange multipliers such that the solution $\mathbf{x}$ is feasible for the original problem and for which the inequality
\begin{equation}
\sum_{i=1}^N \lambda_i (K_i-x_i(1+\sum_{j\neq i}^N a_{ij}x_j)) \geq 0
\end{equation}
is as close to equality as possible. When equality holds the solution $\mathbf{x}$ is the optimal solution to the BQC problem \cite{everett}. 

%Note that this decomposition of the problem into two parts is what lends this approach feasible for a distributed implementation. The first part consist of finding good Lagrange multiplier values, the second solving (\ref{eq:determine-x}) for each user. However, for the moment we proceed with the centralized approach.

\subsection{An Upper Bound}
Before we start with the process of finding the Lagrange multipliers, we provide a result on an upper bound for the optimal value of the original MKP.
\begin{proposition}
Let $\mathbf{x}$ be the solution  and $\mathbf{\lambda}$ the multiplier values of a Lagrange relaxation. Then an upper bound for the optimal value of the BQC problem is given by
\begin{equation}
z^{*} \leq z(\mathbf{x}) + \sum_{i=1}^N \lambda_i (K_i-x_i(1+\sum_{j\neq i}^N a_{ij}x_j)) 
\label{eq:upper-bound}
\end{equation}
\end{proposition}
Proof of this proposition can be found in \cite[Thm 4.1]{magazine1984319}.

\subsection{Finding the Lagrange Multipliers}
Clearly the quality of the solution depends on the Lagrange multipliers, as a low value of the summation in (\ref{eq:upper-bound}) provides a small gap to the optimal value. Finding Lagrange multipliers for integer programming have been investigated in e.g. \cite{everett} and \cite{SenjuandToyoda}, and specifically for the MKP in \cite{magazine1984319}. We adopt the algorithm proposed in \cite{magazine1984319}, by changing it to our modified version of the MKP.

% and \cite{volgenant}. We adopt the algorithm proposed in \cite{volgenant}. This algorithm is an extension of the one proposed in \cite{magazine1984319} and was shown to both give lower upper bounds of the optimal value and also better feasible solutions.

The algorithm consists of the following steps:\newline
\textit{Step 0 (intialize and normalize)}\newline
Let $\lambda_i = 0$ and $x_i = 1$, $i=1,...,N$. \newline
Normalize the coefficients:\begin{align}a_{ij}=& a_{ij}/K_j \hspace{0.2cm}\text{for }j=1,...,N \text{ and }i=1,...,N \nonumber\\
a_{ii} =& 1/K_i \hspace{0.2cm}\text{for }i=1,...,N\nonumber \\
b_i =& 1 \hspace{0.2cm}\text{for }i=1,...,N\nonumber
\end{align}
Compute $y_i = \sum_{j = 1}^N a_{ij}$ for $i=1,...,N$. \newline \newline
\textit{Step 1 (determine the most violated constraint)}\newline
If $y_i\leq 1$ for all $i$, then $\mathbf{x}$ is feasible and we can stop. Else find the most violated constraint:
\begin{align}
i^{*} &= \arg\max_{i} y_i \nonumber
\end{align}
\textit{Step 2 (compute the increase of the multipliers $\lambda_{j^1},...,\lambda_{j^l}$)}
Compute
\begin{equation}
j^{*} = \arg\max_j \frac{a_{i^{*}j}}{c_j}x_j \nonumber
\end{equation}
\textit{Step 3 (increase the $\lambda$s)}\newline
Set $\lambda_{j^{*}} = \lambda_{j^{*}}+a_{i^{*}j^{*}}$\newline
Set $x_{j^{*}} = 0$, $y_{j^{*}} = 0$ and $y_i = y_i-a_{ij^{*}}$.\newline
If $y_i\leq 1$ for all $i$, then go to step 4; otherwise go to step 1.\newline 

The idea of the algorithm is as follows: find the most violated constraint and find the variable that contributes the most to violating this constraint as a fraction of the revenue it would bring the spectrum holder. Set this variable to zero along with its constraint, and reduce all other constraints by the contribution this variable had on these constraints.

\begin{algorithm}[t]
\caption{Channel Selection}
\label{algo:1}
%\small
\begin{algorithmic}[1]
\STATE Consider any $\mathbf{x}$ that is a solution to (\ref{eq:max-z})-(\ref{eq:x-0-1}).
\STATE For each user $i$ with $x_i = 1$ select a random channel from $K_i$.
\FOR{\textbf{ each} user i with $x_i = 1$}
\STATE Compute $\omega_i^k = \sum_{j\neq i} a_j^k I_{j,i}$.
\STATE Select $k^{*} = \arg\min_k \omega_i^k$
\ENDFOR
\STATE Repeat 3-6 until no more adjustments can be performed.
\end{algorithmic}
\end{algorithm}

\section{Channel Selection}
\label{sec:channel-selection}
Through the heuristic algorithm presented in the previous two sections we have found the users which are allowed to transmit in order to achieve a feasible solution of users that can transmit simultaneously while achieving their SINR requirements. However, the solution to this problem does not actually find the channels to be used by each user with $x_i = 1$. In a single-channel network where power control is used to achieve SINR targets at different users, if a feasible solution exists a simple greedy distributed algorithm converges \cite{Foschini1993}. Unfortunately, such a result does not exist for a system consisting of multiple channels.

In the general case this is not a trivial problem. However, by assuming that the channel gains between users are reciprocal (such that $g_{i,j} = g_{j,i}$), a simple greedy algorithm is guaranteed to converge. Such an algorithm is given in Algorithm \ref{algo:1}. 

\begin{proposition}
Algorithm \ref{algo:1} will converge.
\label{prop:2}
\end{proposition}
\begin{IEEEproof}
See Appendix \ref{app:proof-prop2}
\end{IEEEproof}
%From the proof of Proposition \ref{prop:2} it is evident that an AP always prefers unblocked channels compared to blocked ones. By the fact that Proposition \ref{prop:1} guarantees the existence of an unblocked channel for each AP with $x_i=1$ and that Algorithm \ref{algo:1} converges we conclude that Algorithm \ref{algo:1} converges to a successful allocation. 

An example of a scenario where such an assumption can hold is wireless APs such as Wi-Fi networks. In such a network, the APs have about the same coverage area and thus the channel gain between AP $i$ and AP $j$ is approximately the same as between AP $j$ and AP $i$.

\section{Performance Evaluation}
\label{sec:performance}
Through this paper we have presented a heuristic approach to solving the simultaneous SINR problem. This process consisted of two main steps: simplifying the constraints in the original problem to binary quadratic constraints (BQCs) and using a heuristic algorithm for solving the BQCs. We are thus interested in evaluating how the transformation from the original non-convex constraints to BQCs affects the optimal value as well as how the heuristic performs compared to the optimal value of the BQCs. Also, Proposition \ref{prop:1} says that the solution obtained from the maximization problem (\ref{eq:max-z})-(\ref{eq:x-0-1}) is a feasible solution almost surely for $N\rightarrow \infty$. Thus, for finite values of $N$ we are interested in how many of the users with $x_i = 1$ actually achieve their SINR requirements when Algorithm \ref{algo:1} terminates. These issues are investigated in this section as well as the impact of a distributed implementation of the access control algorithm.

\begin{table}[t]
\centering
\caption{Set of SINR targets and their revenue value for the two sub problems, (i) maximizing number of satisfied users and (ii) maximizing revenue}
\begin{tabular}{c|c|c|c|c|c}
\hline 
SINR Targets & 0 dB & 3 dB & 6 dB & 9 dB & 12 dB \\ \hline \hline
$c$ Value (Max Sat Users) & 1 & 1 & 1 & 1 & 1 \\ \hline
$c$ Value (Max Revenue) & 1 & 2 & 3 & 4 & 5 \\ \hline \hline
\end{tabular}
\label{tab:sinr-targets}
\end{table}

\subsection{Analytical Analysis of Geometric Signal Propagation}
We start by a theoretical analysis of the transformation from the problem in (\ref{eq:org-prob})-(\ref{eq:a-0-1}) to the one given in (\ref{eq:max-z-t})-(\ref{eq:x-0-1-t}). The constraint given in (\ref{eq:linearized-c-t}) is a sufficient condition for a feasible solution. If channel gain can be modeled as proportional to distance so that the triangular inequalities hold\footnote{let $x,y,z$ be points in a Euclidean space, then $d(x,y)<d(x,z)+d(y,z)$ hold.}, and all users have the same SINR requirement, we can through a necessary condition bound the optimal solution as follows.
\begin{proposition}
Assume the optimal solution to (\ref{eq:org-prob})-(\ref{eq:a-0-1}) can achieve $OPT$. Then the optimal solution to (\ref{eq:max-z-t})-(\ref{eq:x-0-1-t}), $OPT'$, can achieve
\begin{equation}
\frac{OPT}{\min\{2^{\alpha}-1,10\}}-1\leq OPT'
\end{equation}
where $\alpha$ is the pathloss exponent.
\label{prop:G}
\end{proposition}
\begin{proof}See Appendix \ref{app:proof-propG}\end{proof}
Thus, under these conditions we can find a constant factor approximation of the optimal solution. However, if received signal power don't satisfy the triangular inequalities it has been shown that it is NP-hard to approximate the optimal solution of a scheduling with one channel to within $N^{1-\eta}$ for $\eta>0$ \cite[Theorem 6.1]{5062108}. As adding channels increases the decision complexity, this also holds for our problem.

\begin{figure}[t]
\centering
\includegraphics[width = 0.8\columnwidth]{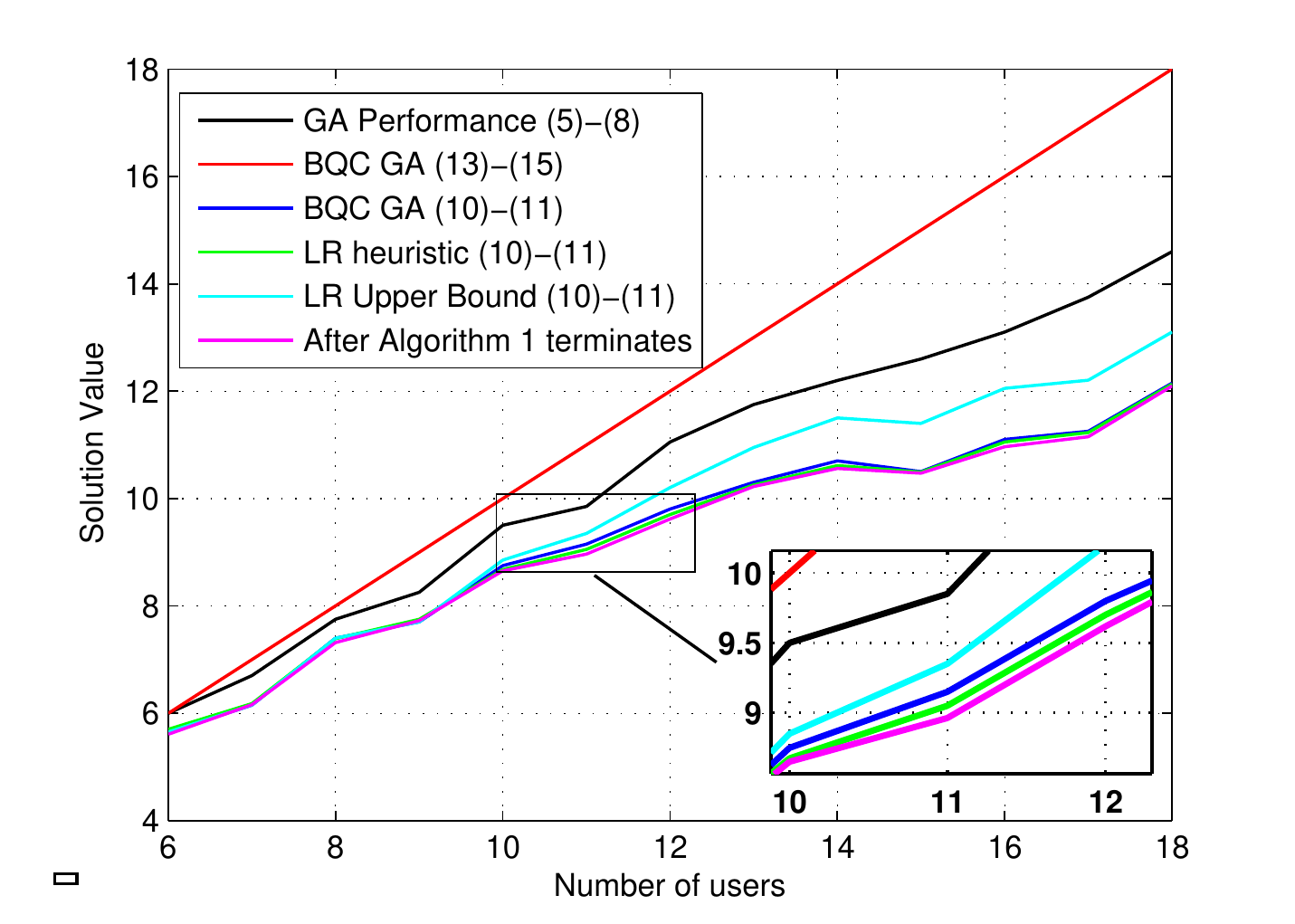}
\caption{Solution of the original problem (GA) compared to the GA solution value of the problem given in (13)-(15) as well as the solution value of the BQC problem (10)-(11) (GA, heuristic, upper bound and channel allocation) for the max sat problem.
}
\label{fig:opt-value-nona}
\end{figure}

\begin{figure}[t]
\centering
\includegraphics[width = 0.8\columnwidth]{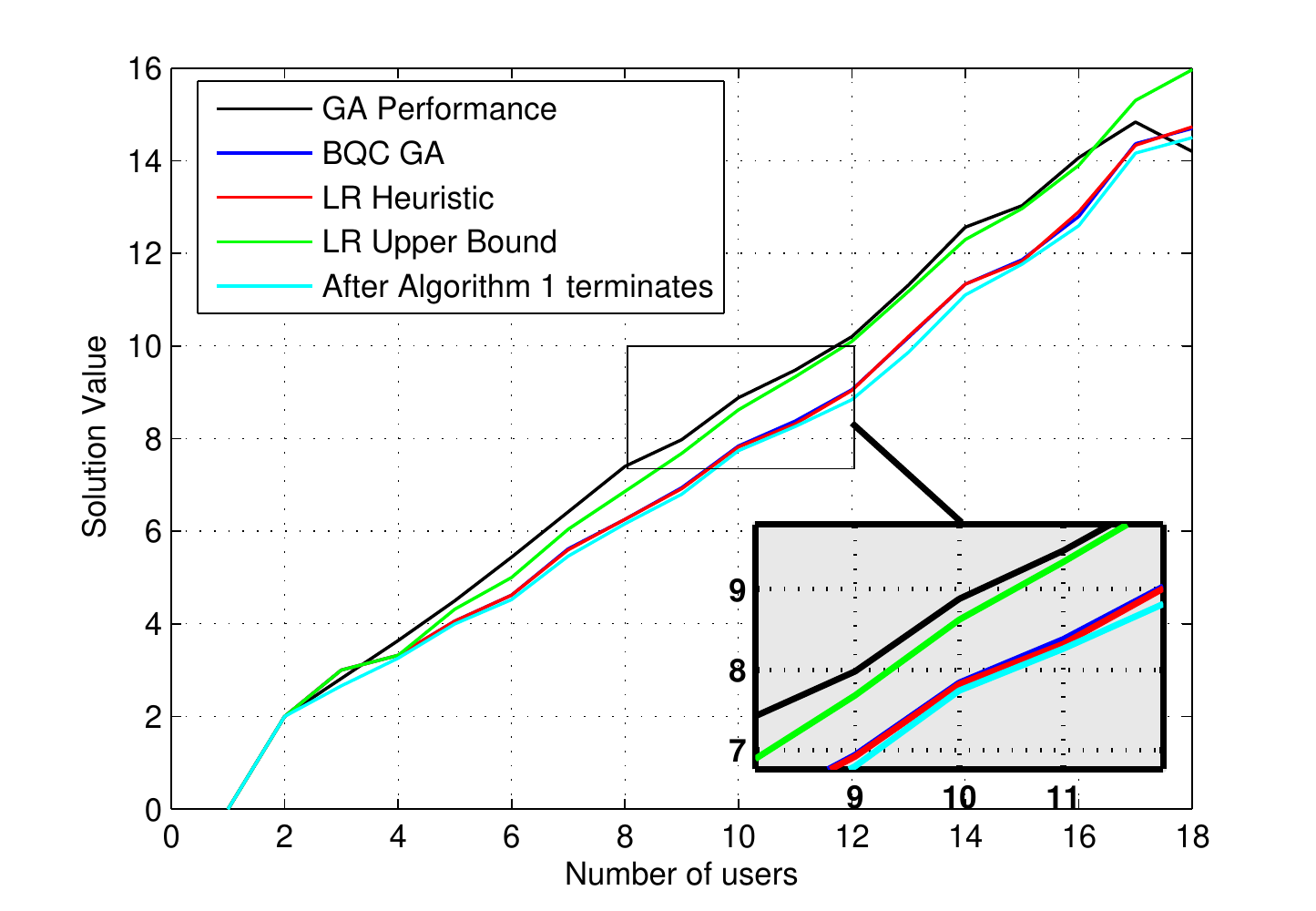}
\caption{Solution of the original problem (GA) compared to the solution value of the BQC problem (13)-(15) (GA, heuristic, upper bound and channel allocation) for the max sat problem.}
\label{fig:opt-value}
\end{figure}

\subsection{Simulation Setup}
To evaluate the issues mentioned above we simulate an environment of constant user density. Specifically, we have a density of 1/800 [users/m$^2$]. The rest of the user parameters are set so that the optimal value of the original problem is less than the number of users for most instances. We create the environment by positioning the users randomly in a square according to a uniform distribution in two dimensions. The distance from a user to its receiver is given by Gaussian distribution with mean value 10 and variance 5. The transmit power of each user is set to 1 W and the noise variance is set to $10^{-8}$ W. 

The SINR requirement of a specific user is chosen at random from a set of SINR targets, which is given in Table \ref{tab:sinr-targets}. We consider two sub problems, (i) maximizing number of satisfied users, in which case the revenue of each user is equal, and (ii) maximizing revenue, in which case satisfying larger SINR targets leads to a larger revenue.

%\begin{table}
%\centering
%\caption{Set of SINR targets and their revenue value for the two sub problems, (i) maximizing number of satisfied users and (ii) maximizing revenue}
%\begin{tabular}{c|c|c}
%\hline 
%SINR Targets & $c$ Value (Max Sat Users) & $c$ Value (Max Revenue) \\ \hline \hline
%0 dB & 1 & 1 \\ \hline
%3 dB & 1 & 2 \\ \hline
%6 dB & 1 & 3 \\ \hline
%9 dB & 1 & 4 \\ \hline
%12 dB & 1 & 5 \\ \hline
%\end{tabular}
%\label{tab:sinr-targets}
%\end{table}

%The last parameter of the users is the available channels. The maximum number of available channels to any user is uniformly distributed between $2$ and the number of users. The specific number of available channels to a user is then uniformly distributed between $2$ and this maximum number. E.g. if the maximum number of available channels is 6, the number of available channels to a user is uniformly distributed between 2 and 6. The reason why the minimum number of channels is set to $2$ is because from (\ref{eq:linearized-c}) we see that if $K_i = 1$, $x_i = 0$ regardless of the environment. The difference in performance between $U(1,\max)$ and $U(2,\max)$ is shown in Fig. \ref{fig:min-c}.

%\begin{figure}[t]
%\centering
%\includegraphics[width = 0.8\columnwidth]{percent_of_ga_users.pdf}
%\caption{Performance as percentage of GA solution value.}
%\label{fig:ga-percent}
%\end{figure}

To provide a reference value to the optimal value of the simplified problem compared to the original optimization problem, we solve (\ref{eq:org-prob})-(\ref{eq:a-0-1}) using a genetic algorithm (GA) from the MATLAB Global Optimization Toolbox. GAs are particularly suitable when the structure of the solution space of the problem is unknown and little theoretic analysis of the problem has been done. To benchmark the quality of the heuristic algorithm, we also solve the constraint transformed problems through the GA function in MATLAB. The population size of the GA was set to ten times the number of decision variables.
\begin{figure}
\centering
\includegraphics[width = 0.8\columnwidth]{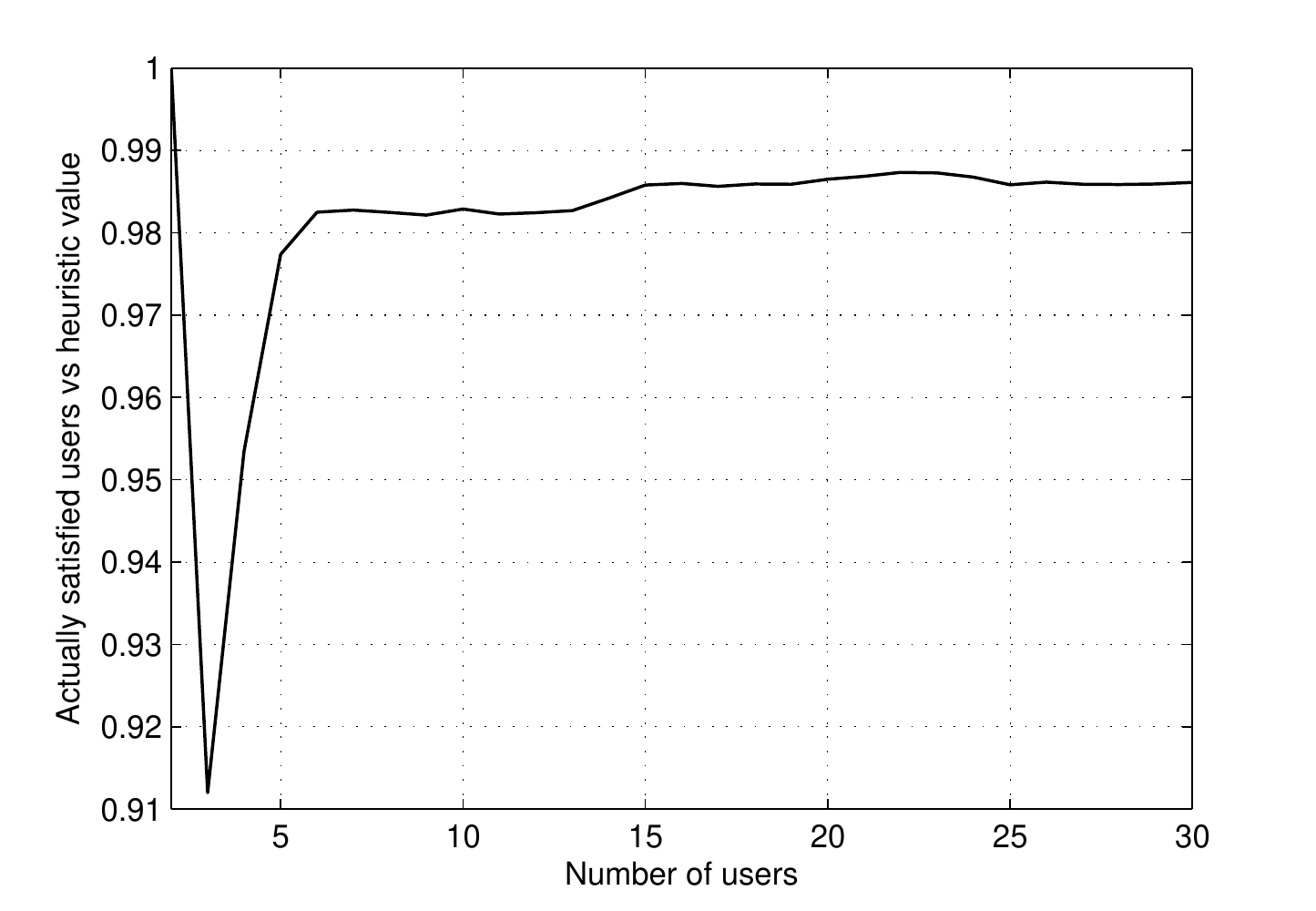}
\caption{Actually satisfied users after Algorithm 1 terminates vs solution value of the heuristic.}
\label{fig:sat-vs-z}
\end{figure}

\subsection{Simulation Results}
Fig. \ref{fig:opt-value-nona} shows the solution values for an environment with $5$ available channels where the goal is to maximize the number of satisfied users. Since all users have the same spectral resources, the constraint transformation in (\ref{eq:linearized-c}) will not hold and as can be seen from the results the solution values obtained from (\ref{eq:max-z})-(\ref{eq:x-0-1}) are not valid solutions. On the other hand, the solutions obtained from the constraint transformation in (\ref{eq:linearized-c-t}) are valid solutions, but we can see that the constraint in (\ref{eq:linearized-c-t}) is only a sufficient condition for a valid solution as the gap between the GA performance of the original problem in (\ref{eq:org-prob})-(\ref{eq:a-0-1}) and the performance of (\ref{eq:max-z-t})-(\ref{eq:x-0-1-t}) is roughly 2 for high numbers of users.

For the results in Fig \ref{fig:opt-value} - \ref{fig:d-vs-c-no} we let the spectral resources at each user vary such that its available channels is drawn as a random subset of $K_i$ channels from the total number of channels $K$. In the simulations $K_i$ is a uniformly distributed number between $2$ and $K$ and $K$ is set to 10. In Fig. \ref{fig:opt-value} the solution value of the different algorithms are plotted for the max sat problem. For 2 and 3 users, the BQCs optimal value can actually be greater than that of the GA because when the number of channels at any user is at least 2, the BQC problem is trivial for 2 users and will always set $x = 1$, as can be seen from (\ref{eq:linearized-c}). For large values of $N$ we see that the solutions obtained by BQC problem can surpass the GA solution of the original problem. A possible explanation for this is that with increased number of users the solution space for the original problem is so large that it becomes less likely that the GA will find the global optimum. With this in mind, the results indicate that the constraint transformation captures the performance of the system well.

In Fig. \ref{fig:sat-vs-z}, the difference between the optimal value of the heuristic and the number of users actually achieving their SINR requirements when Algorithm \ref{algo:1} terminates is given, normalized against the former. This further validates the constraint approximation done in this paper, as for number of users greater than 5 there is a less than 2\% difference between the optimal value of the heuristic and the number of users that actually achieve their SINR. Note that in the simulations here we have not assumed reciprocal channel gains (as was assumed in the proof of Proposition \ref{prop:2}), and some of the gap can be attributed to this fact.

\begin{figure}[t]
\centering
\includegraphics[width = 0.8\columnwidth]{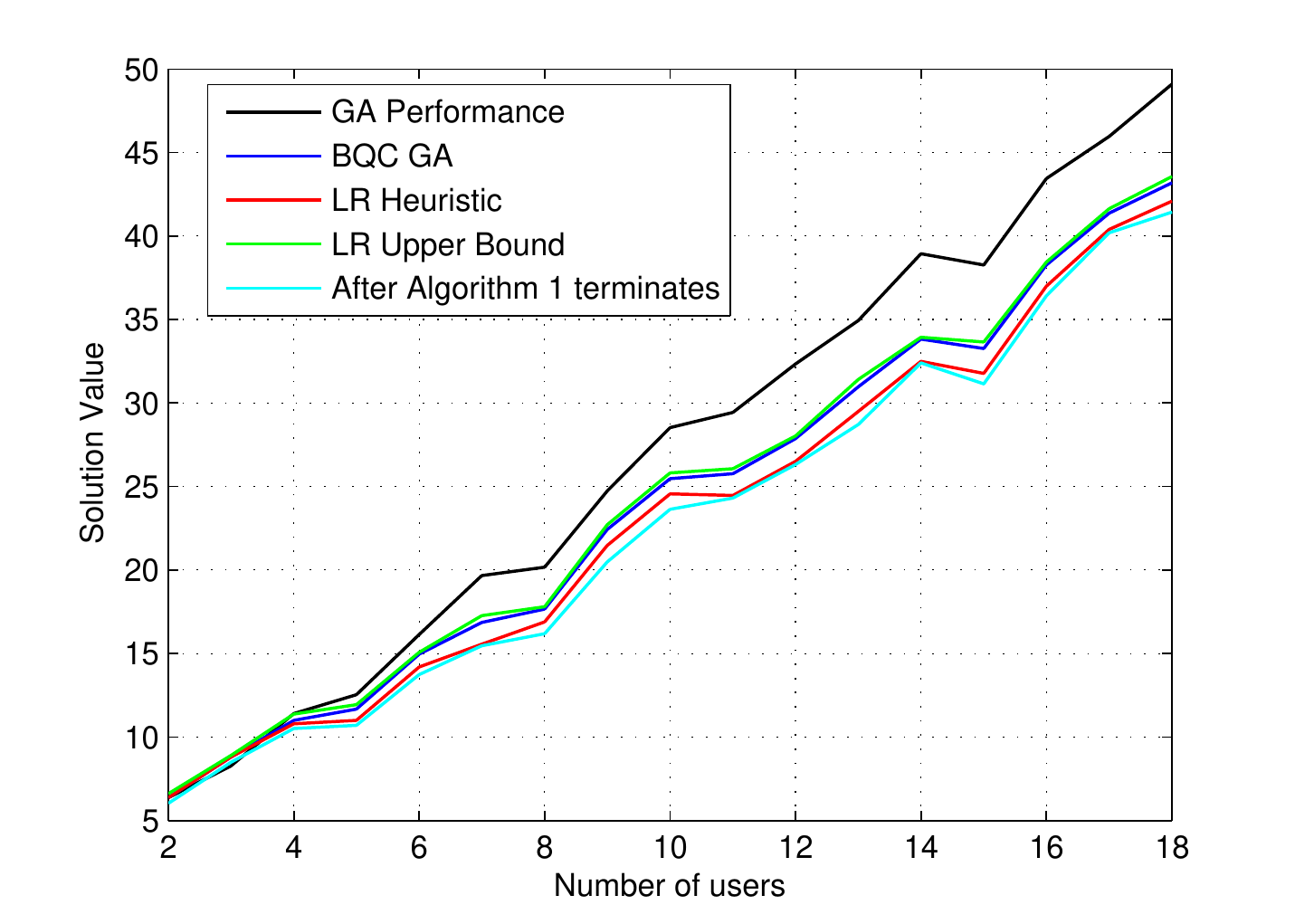}
\caption{Solution of the original problem (GA) compared to the solution value of the BQC problem (13)-(15) (GA, heuristic, upper bound and channel allocation) for the max revenue problem.}
\label{fig:opt-value-revenue}
\end{figure}

In the previous results the goal has been to maximize the number of satisfied users. In Fig. \ref{fig:opt-value-revenue} the goal is to maximize the revenue, where each users SINR requirement yields a corresponding profit as given in Table \ref{tab:sinr-targets}. Again, the constraint transformation performs within the same percentage of the GA solution of the original problem, between 10-15\% for 10-18 users.

%In Fig. \ref{fig:min-c} the difference between setting minimum number of channels a user can have to 1 compared to 2 is given. As can be seen the difference is the largest for small number of users, which is expected from analyzing the BQC constraint. 

\begin{figure}
\centering
\includegraphics[width = 0.8\columnwidth]{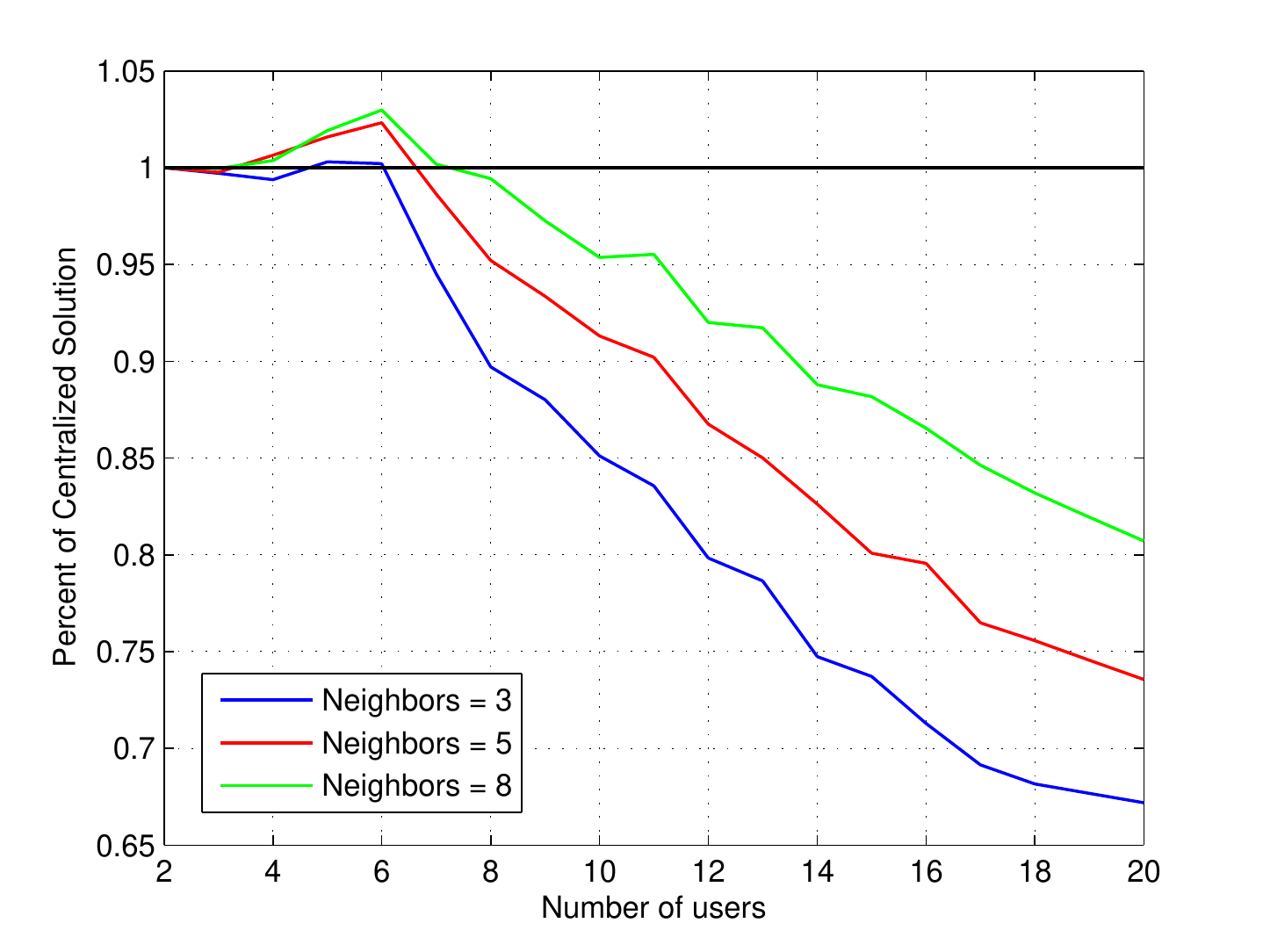}
\caption{	Distributed performance as percentage of centralized solution after Algorithm 1 terminates.}
\label{fig:d-vs-c-no}
\end{figure}

%\begin{figure}
%\centering
%\includegraphics[width = 0.8\columnwidth]{percent_of_centralized_w_penelty_users.pdf}
%\caption{Distributed performance as percentage of centralized solution with power penalty.}
%\label{fig:d-vs-c-w}
%\end{figure}

In Fig. \ref{fig:d-vs-c-no} we investigate the effect of taking only a subset of interfering users into account at each user. As our sufficient condition is based on bounding the size of different sets of interfering users, limiting the number of users that are contributing to the interference at each user will increase the number of users the heuristic algorithm allows to transmit. The question if how this affects the actual channel allocation based on the users that are allowed to transmit from constraint transformation. In Fig. \ref{fig:d-vs-c-no} we have plotted the number of users that actually achieve their target SINR after Algorithm \ref{algo:1} terminates for different neighbor degrees, as percentage of the number of users that achieve their SINR targets after Algorithm \ref{algo:1} terminates with global knowledge at the heuristic. Based on the user density we have calculated the distance within which $x$ neighboring users are expected to be found. Thus, for Neighbors = $x$, on average $x$ neighbors will be taken into account when solving the optimization problem. As can be seen from Fig. \ref{fig:d-vs-c-no}, for a low number of users limiting the number of neighbors taken into account actually increases the number of satisfied users even after Algorithm \ref{algo:1} terminates. Thus, we can assume that the sufficient condition is a bit restrictive. However, for increasing number of users limiting the number of neighbors that are taken into account severely degrades performance. This result limits the sufficient condition to be used in a distributed manner.

%the performance of the distributed implementation as percentage of the centralized solution is given without and with power penalty. Without a power penalty, the more knowledge each user has the closer the solution is to the centralized solution. However, as can be seen by Fig. \ref{fig:d-vs-c-w}, if we account for the power spent obtaining this information this is not the case. For a small number of users information about neighboring users is less important as the competition for the resources is less critical. However, for large number of users it is more beneficial to have more information in order for the solution to approach the centralized solution. This can be seen as for knowledge of 8 neighbors the performance is almost constant, whereas by knowing 6 neighbors the performance surpasses that of 3 neighbors at 19 users.

The result that capture the benefit of the heuristic algorithm the most is given in Table \ref{tab:time}, which shows the time taken to solve the problem for the different approaches. As the SINR requirements of the users change with time, the time it takes the algorithm to find a good solution is important. E.g. we see that for 12 users it takes the GA algorithm 30 minutes to solve the optimization problem. In this time it is likely that some of the users transmission parameters have changed, such as the SINR requirement and worst case signal strength, rendering the solution suboptimal.

Due to the time required to solve the original problem through a GA, simulation results are only provided for a maximum of 18 users. However, as our constraint transformation allows us to utilize a heuristic algorithm with a time complexity $\mathcal{O}(2N^2)$ this approach can be used to provide good estimates of the max sat problem and the max revenue problem for large-scale networks.

\begin{table}[t!]
\caption{Time to solve the problem}
\centering
\begin{tabular}{c c c}
\hline
Number of users & Algorithm & Time (s) \\
\hline
\hline
\multirow{3}{*}{2} & GA & 1.982 \\
 & BQC GA & 0.422 \\
 & Heuristic & 0.0001 \\
\hline
\multirow{3}{*}{6} & GA & 125.7 \\
 & BQC GA & 1.072 \\
 & MKP Heuristic & 0.0004 \\
\hline
\multirow{3}{*}{12} & GA & 1832 \\
 & BQC GA & 3.922 \\
 & MKP Heuristic & 0.0009 \\
\hline
\multirow{3}{*}{18} & GA & 16813 \\
 & BQC GA & 7.305 \\
 & MKP Heuristic & 0.0015 \\
\hline
\end{tabular}
\label{tab:time}
\end{table}

\section{Conclusion and Future Work}
\label{sec:conc}
In this paper we have presented a heuristic algorithm for solving the spectrum allocation problem under SINR requirements. Specifically, we investigated how many users that can transmit and achieve their SINR targets simultaneously and maximize the profit to a spectrum holder. As this problem is a non-convex integer problem, we transformed this problem into two binary quadratic constraint problems, one for which the solution is guaranteed to be a feasible solution to the original optimization problem and one which is a feasible solution asymptotically almost surely. To solve the BQC problems, we presented a heuristic algorithm based on Lagrange relaxation which bounds the solution value of the heuristic to the optimal value of the BQC problem. Through simulation results we showed that this approach yields solutions on average at a gap of 10\% from the solutions obtained by a genetic algorithm for the original non-linear problem.

%One problem with our approach is that some users might not be able to transmit to their receivers for a long period of time. Thus, fairness is an important aspect lacking in this paper. This, as well as analytical performance bounds for individual SINR requirements are part of future work. 

\appendices
\input{prop1}
\input{propA}
\input{propConvergence}
\input{propG}

\bibliographystyle{ieeetr}
\bibliography{../../../bib,../../../library}{}
\end{document}

%% file: prop1.tex
\section{Proof of Proposition \ref{prop:1}}
\label{app:proof-prop1}

We need to show that for each $x_i=1$ there exists a channel $k$ such that SINR$_i^k = \frac{S_i}{\sigma^2 + \sum_{j\neq i}a_j^k I_{j,i}}\geq \beta_i$.% almost surely for $N\rightarrow \infty$.

Let $\omega_i^k = \sum_{j\neq i}a_j^k I_{j,i}$ represent the accumulated interference at user $i$ on channel $k$. From the definition of $I_i^{\max}$, it is clear that SINR$_i^k\geq \beta_i$ is equivalent to $\omega_i^k\leq I_i^{\max}$. We define a channel $k$ as blocked if $\omega_i^k > I_i^{\max}$. To prove the proposition we need to prove that if $x_i=1$ (i.e. user $i$ transmits) there exists at least one channel that is not blocked.

We bound the number of blocked channels as follows
\begin{align}
\Phi = &\{k|\omega_i^k>I_i^{\max}\} \nonumber \\
= &\{k|\exists j, a_j^k = 1 \text{ and } I_{j,i}>I_i^{\max}\}\cup \nonumber\\
&\{k|j:a_j^k = 1 \text{ and } I_{j,i}\leq I_i^{\max} \text{ and } \omega_i^k>I_i^{\max}\}\nonumber \\
= &\Phi_1 \cup \Phi_2
\end{align}
$\Phi_1$ contains the channels blocked by strong interferers, i.e. those that can block a band single handedly, whereas $\Phi_2$ contains those channels that are blocked by accumulative interference. We now bound the sizes of $\Phi_1$ and $\Phi_2$. We first bound $\Phi_1$.
\begin{align}
\Phi_1 &= \sum_{k=1}^{K_i}\sum_{j\neq i, I_{j,i}>I_i^{\max}} a_j^k \leq \sum_{j\neq i, I_{j,i}>I_i^{\max}}x_j%\nonumber \\
%&\leq\sum_{j\neq i, I_{j,i}>I_i^{\max}}x_j\text{Prob}(j,k)
\end{align}
%where $\text{Prob}(j,i)$ is the probability that $j$ and $i$ transmits on the same channel.
To bound $\Phi_2$ we use the fact that:
\begin{align}
&\sum_{k\in \Phi_2}\omega_i^k = \sum_{k\in\Phi_2}\sum_{j\neq i} a_j^kI_{j,i} = \sum_{j\neq i, I_{j,i}\leq I_i^{\max}}I_{j,i}\sum_{k\in\Phi_2}a_j^k \nonumber \\
&\leq \sum_{j\neq i, I_{j,i}\leq I_i^{\max}}I_{j,i}x_j
\end{align}
Since for any channel $k\in\Phi_2$ $\omega_i^k>I_i^{\max}$ we have that
\begin{equation}
\sum_{k\in\Phi_2}I_i^{\max}<\sum_{k\in\Phi_2} \omega_i^k \leq \sum_{j\neq i, I_{j,i}\leq I_i^{\max}}I_{j,i}x_j
\end{equation}
Since $I_i^{\max}$ is fixed we have that
\begin{equation}
|\Phi_2|<\frac{1}{I_i^{\max}}\sum_{j\neq i, I_{j,i}\leq I_i^{\max}}I_{j,i}x_j
\end{equation}
And thus that the number of blocked channels is at most
\begin{align}
|\Phi| &= |\Phi_1|+|\Phi_2|<\sum_{j\neq i, I_{j,i}>I_i^{\max}}x_j\nonumber\\
& + \sum_{j\neq i, I_{j,i}\leq I_i^{\max}}\frac{I_{j,i}}{I_i^{\max}}x_j \nonumber \\
&= \sum_{j\neq i} x_j\frac{I^{+}_{j,i}}{I_i^{\max}}% \sim \sum_{j\neq i} x_j\frac{I^{+}_{j,i}}{I_i^{\max}}\frac{|\mathcal{K}_j\cap\mathcal{K}_i|}{K_jK_i}
\end{align}
And thus if $\sum_{j\neq i} x_j\frac{I^{+}_{j,i}}{I_i^{\max}}\leq K_i-1$ there exists at least one channel on which user $i$ can achieve its SINR requirement.\qed

%% file: propA.tex
\section{Proof of Proposition \ref{prop:a}}
\label{app:proof-propA}

Consider the sets $\Phi_1$ and $\Phi_2$ defined in Appendix \ref{app:proof-prop1}. We now bound them again. We bound $\Phi_1$ as
\begin{align}
\Phi_1 &= \sum_{k=1}^{K_i}\sum_{j\neq i, I_{j,i}>I_i^{\max}} a_j^k \leq \sum_{j\neq i, I_{j,i}>I_i^{\max}}x_j\text{Prob}(j,i)%\nonumber \\
%&\leq\sum_{j\neq i, I_{j,i}>I_i^{\max}}x_j\text{Prob}(j,k)
\end{align}
where $\text{Prob}(j,i)$ is the probability that $j$ will contribute to blocking on any of $i$'s channels. We note that the same can be done for $\Phi_2$ so that
\begin{align}
&\sum_{k\in \Phi_2}\omega_i^k \leq \sum_{j\neq i, I_{j,i}\leq I_i^{\max}}I_{j,i}x_j\text{Prob}(j,i)
\end{align}
and
\begin{equation}
\sum_{k\in\Phi_2}I_i^{\max}<\sum_{k\in\Phi_2} \omega_i^k \leq \sum_{j\neq i, I_{j,i}\leq I_i^{\max}}I_{j,i}x_j\text{Prob}(j,i)
\end{equation}
so that
\begin{equation}
|\Phi_2|<\frac{1}{I_i^{\max}}\sum_{j\neq i, I_{j,i}\leq I_i^{\max}}I_{j,i}x_j\text{Prob}(j,i)
\end{equation}
The total number of blocked channels can now be given as
\begin{align}
|\Phi| = \sum_{j\neq i} x_j\frac{I^{+}_{j,i}}{I_i^{\max}}\text{Prob}(j,i) \sim \sum_{j\neq i} x_j\frac{I^{+}_{j,i}}{I_i^{\max}}\frac{|\mathcal{K}_j\cap\mathcal{K}_i|}{K_j}
\end{align}
Clearly as $N\rightarrow \infty$
\begin{equation}
\sum_{j\neq i} x_j\frac{I^{+}_{j,i}}{I_i^{\max}}\text{Prob}(j,i)-\sum_{j\neq i} x_j\frac{I^{+}_{j,i}}{I_i^{\max}}\frac{|\mathcal{K}_j\cap\mathcal{K}_i|}{K_j}\rightarrow 0
\end{equation}
due to the central limit theorem. Thus it follows that if $\sum_{j\neq i} x_j\frac{I^{+}_{j,i}|\mathcal{K}_j\cap\mathcal{K}_i|}{I_i^{\max}K_j}\leq K_i-1$, there exists at least one channel on which user $i$ can achieve its SINR requirement almost surely.\qed

%% file: propConvergence.tex
\section{Proof of Proposition \ref{prop:2}}
\label{app:proof-prop2}
We prove that a greedy algorithm such as the one given in Algorithm \ref{algo:1} will converge under the assumption of reciprocal
channel gains (i.e. $g_{ij} = g_{ji}$ is the pathloss between user $i$ and $j$).

Consider a utility function $u_i^{k}$ for each user $i$ that depends on $k$ as follows
\begin{equation}
u_i^{k} = \sum_{j\neq i}a_i^k a_j^k g_{ji}P_jP_i.
\end{equation}
Clearly minimizing $\omega_i^k$ in Algorithm \ref{algo:1} is equivalent to minimizing $u_i^k$. Consider a potential function $P(\mathbb{A})$ defined for an allocation $\mathbb{A}$ as
\begin{equation}
P(\mathbb{A}) = \sum_{i=1}^N\sum_{k=1}^{K_i} u_i^k
\end{equation}
We prove convergence by showing that the game forms a generalized ordinal potential game. Convergence is then guaranteed as all finite generalized ordinal potential games have a pure strategy equilibrium \cite[Corollary 2.2]{Monderer1996} and all finite generalized ordinal potential games have a finite improvement path \cite[Lemma 2.3]{Monderer1996}.

Now, user $i$ will change its allocation from channel $k$ to $k'$ given that
\begin{equation}
\omega_i^{k'} = \sum_{j\neq i}a_{j}^{k'}g_{ji}P_j<\sum_{j\neq i}a_{j}^{k}g_{ji}P_j = \omega_i^{k}
\end{equation}
By multiplying each side by $P_i$ we have
\begin{equation}
u_i^{k'} =\sum_{j\neq i}a_{j}^{k'}g_{ji}P_jP_i<\sum_{j\neq i}a_{j}^{k}g_{ji}P_jP_i = u_{i}^k
\end{equation}
Assume before user $i$ changed from $k$ to $k'$ that user $l$ transmits on $k$ and user $m$ transmits on $k'$. Since user $i$ no longer transmits on $k$, clearly $\omega_l^k \text{ (after)}<\omega_l^k \text{ (before)}$. %However, the same is not true for AP $m$. Instead we must show that
%\begin{align}
%&\sum_{i=1}^N\sum_{j\neq i} a_i^ka_j^kg_{ji}P_jP_i + \sum_{i=1}^N\sum_{j\neq i} a_i^{k'}a_j^{k'}g_{ji}P_jP_i \text{ (after)} \nonumber \\
%&<\sum_{i=1}^N\sum_{j\neq i} a_i^ka_j^kg_{ji}P_jP_i + \sum_{i=1}^N\sum_{j\neq i} a_i^{k'}a_j^{k'}g_{ji}P_jP_i \text{ (before)}\nonumber
%\end{align}
Define $\Delta u_i$, $\Delta u_l$ and $\Delta u_m$ as
\begin{equation*}
\Delta u_i = u_{i}^{k} - u_{i}^{k'} = \sum_{j\neq i}a_{j}^{k}g_{ji}P_jP_i-\sum_{j\neq i}a_{j}^{k'}g_{ji}P_jP_i>0,
\end{equation*}
\begin{equation*}
\Delta u_l = u_l(\text{before $i$ changed}) - u_l(\text{after $i$ changed}) = a_{l}^{k}g_{il}P_lP_i
\end{equation*}
and
\begin{align}
\Delta u_m &= u_m(\text{before $i$ changed}) - u_m(\text{after $i$ changed}) \nonumber \\
&= -a_{m}^{k'}g_{im}P_mP_i \nonumber
\end{align}
$\Delta u_l$ corresponds to the users that have gained by user $i$'s change from channel $k$ to $k'$ in terms of increased SINR. $\Delta u_m$ corresponds to the users that have lost by user $i$'s change from channel $k$ to $k'$ in terms of decreased SINR. %Due to reciprocal channel gains we see that
The total change over all affected users is thus
\begin{align}
&\Delta P(\mathbb{A},\mathbb{A'}) = \Delta u_i + \sum_{l\neq i} \Delta u_l + \sum_{m\neq i} \Delta u_m \nonumber \\
&=\sum_{j\neq i}a_{j}^{k}g_{ji}P_jP_i-\sum_{j\neq i}a_{j}^{k'}g_{ji}P_jP_i + \sum_{j\neq i}a_{j}^{k}g_{ji}P_jP_i \nonumber \\
&-\sum_{j\neq i}a_{j}^{k'}g_{ji}P_jP_i \nonumber \\
&= 2\biggl(\sum_{j\neq i}a_{j}^{k}g_{ji}P_jP_i-\sum_{j\neq i}a_{j}^{k'}g_{ji}P_jP_i\biggr)>0
\end{align}
Thus, thus if $u_i^k>u_i^{k'}$ $\Rightarrow$ $P(\mathbb{A})>P(\mathbb{A'})$ and this is thus a generalized ordinal potential game.\qed
%\begin{equation*}
%u_{i}^{k} = \sum_{l\neq i} \Delta u_l \hspace{0.5cm}\text{and}\hspace{0.5cm}u_i^{k'} = \sum_{m\neq i} \Delta u_m.
%\end{equation*}
%Thus we have that
%\begin{equation}
%\Delta u_i + \sum_{l\neq i} \Delta u_l + \sum_{m\neq i} \Delta u_m
%\end{equation}

%\begin{equation}
%\delta_i+\sum_{l\neq i}\delta_l -\sum_{m\neq i} \delta_m>0
%\end{equation}
%and that $F(\mathbb{A})$ strictly decreases after a change and will thus converge.\qed

%% file: propG.tex
\section{Proof of Proposition \ref{prop:G}}
\label{app:proof-propG}

%We now prove the necessary condition.
%
%Let $\omega_i^{k+} = \sum_{j\neq i}a_j^k I_{j,i}^{+}$. We then have
%\begin{align}
%&\sum_{j\neq i} x_jI_{j,i}^{+} = \sum_{j\neq i}\sum_{k=1}^K a_j^k I_{j,i}^{+} = \sum_{k=1}^K \sum_{j\neq i} a_j^k I_{j,i}^{+}  \nonumber \\ 
%&= \sum_{k=1}^K \omega_i^{k+} = \sum_{k=1,a_i^k = 1}^K \omega_i^{k+} + \sum_{k=1,a_i^k = 0}^K \omega_i^{k+} \nonumber \\
%&\leq I^{\max} + \sum_{k=1,a_i^k = 0}^K \omega_i^{k+}
%\end{align}
%
%Let user $m$ be the user closest user $i$ that transmits on $k$. Then
%\begin{equation}
%\omega_m^{k+}\leq I^{\max} \Rightarrow \sum_{j\neq m}a_j^k I_{j,m}^{+}\leq I^{\max}
%\end{equation}
%By geometry we now have that any other user allocated channel $k$ must satisfy $d_{j,i}\geq \frac{d_{j,m}}{2}$, since user $m$ is the closest one to $i$. It then follows that $\frac{2^{\alpha}}{d_{j,m}^{\alpha}}\geq \frac{1}{d_{j,i}^{\alpha}}$ and therefore that $I_{j,i}^{+}\leq 2^{\alpha}I_{j,m}^{+}$.
%
%We now have that
%\begin{align}
%&\omega_i^{k+} = \sum_{j\neq i}a_j^k I_{j,i}^{+} = I_{m,i}^{+} + \sum_{j\neq i,m}a_j^k I_{j,i}^{+} \nonumber \\
%&\leq I^{\max} + \sum_{j\neq i,m}a_j^k 2^{\alpha}I_{j,k}^{+} \leq (2^{\alpha}+1)I^{\max}
%\end{align}

To prove the proposition we first introduce a necessary condition for a successful spectrum allocation:
\begin{lemma}
Let $\mathbf{x}$ be a successful spectrum allocation. Then for each $x_i$ the following holds
\begin{equation}
x_i\bigl(1+\sum_{j\neq i} x_j \frac{I_{j,i}^{+}}{I^{\max}}\bigr) \leq \min(2^{\alpha}+1,10)K_i
\label{eq:nec}
\end{equation}
\end{lemma}
The proof is similar to the proof of Lemma 3 in \cite{optimus}.

By comparing the necessary condition of (\ref{eq:nec}) with the sufficient condition in (\ref{eq:linearized-c-t}) we can bound the optimal value. Let $\frac{I_{i,j}^{+}}{I^{\max}}$ be denoted as $a_{j,i}$ and $\min(2^{\alpha}+1,10)$ as $C$. We now define two sets:
\begin{align}
&\mathcal{L}_s = \{i |x_i = 1, x_i(1+\sum_{j\in \mathcal{L}_s,j\neq i}x_ja_{j,i})\leq K_i\} \\
&\mathcal{L}_n = \{i |x_i = 1, x_i(1+\sum_{j\in \mathcal{L}_n,j\neq i}x_ja_{j,i})\leq CK_i\}
\end{align}
Thus, $|\mathcal{L}_s|$ is the number of users satisfying the sufficient constraints and $|\mathcal{L}_n|$ is the number of users satisfying the necessary constraints. Clearly the optimal value (OPT) of (\ref{eq:org-prob})-(\ref{eq:a-0-1}) is bounded as $|\mathcal{L}_s|\leq OPT\leq |\mathcal{L}_n|$.

Let $\epsilon_i$ be smallest $a_{i,j}$ in the sum over $\mathcal{L}_n$, i.e. $\epsilon_i = \min_{j\in \mathcal{L}_n}a_{j,i}$. Then we have
%\begin{align}
%&(|\mathcal{L}_n|-1)\epsilon_i\leq \sum_{j\in \mathcal{L}_n,j\neq i}x_ja_{j,i} \leq CK_i-1 \nonumber \\
%\Rightarrow & \epsilon_i\leq \frac{CK_i-1}{|\mathcal{L}_n|-1} \nonumber 
%\end{align}
%which maximizes $|\mathcal{L}_n|$. Consider
%\begin{align}
%&1+ \sum_{j\in \mathcal{L}_n,j\neq i}x_ja_{j,i} = 1+\sum_{j\in \mathcal{L}_s,j\neq i}x_ja_{j,i}\nonumber\\
%&+\sum_{j\in \mathcal{L}_n\backslash\mathcal{L}_s,j\neq i}x_ja_{j,i} \geq 1+\sum_{j\in \mathcal{L}_s,j\neq i}x_ja_{j,i}\nonumber\\
%&+(|\mathcal{L}_n\backslash\mathcal{L}_s|-1)\epsilon_i
%\end{align}
%Thus we have
%\begin{align}
%&1+\sum_{j\in \mathcal{L}_s,j\neq i}x_ja_{j,i}+(|\mathcal{L}_n\backslash\mathcal{L}_s|-1)\epsilon_i \nonumber \\
%&\leq K_i + (|\mathcal{L}_n\backslash\mathcal{L}_s|-1)\epsilon_i \leq CK_i
%\end{align}
%Summation over all $i$ yields
%\begin{align}
%&\sum_i (K_i+(|\mathcal{L}_n\backslash\mathcal{L}_s|-1)\epsilon_i)\leq C\sum_iK_i \nonumber \\
%&\sum_i K_i+\frac{(|\mathcal{L}_n\backslash\mathcal{L}_s|-1)}{|\mathcal{L}_n|-1}\sum_i(CK_i-1)\leq C\sum_iK_i \nonumber \\
%&1+C\frac{(|\mathcal{L}_n\backslash\mathcal{L}_s|-1)}{|\mathcal{L}_n|-1}-\frac{|\mathcal{L}_n|}{\sum_i K_i}\leq C \nonumber \\
%&1 + (|\mathcal{L}_n\backslash\mathcal{L}_s|-1)\frac{C}{|\mathcal{L}_n|}\leq C
%\end{align}
%
%
%Since $\mathcal{L}_s$ is a subset of $\mathcal{L}_n$, $|\mathcal{L}_n\backslash\mathcal{L}_s| = |\mathcal{L}_n|-|\mathcal{L}_s|$, and we get
%\begin{equation}
%|\mathcal{L}_n|\leq C(|\mathcal{L}_s|+1)
%\end{equation}\qed

\begin{align}
&|\mathcal{L}_n|\epsilon_i\leq 1+ \sum_{j\in \mathcal{L}_n,j\neq i}x_ja_{j,i} \leq CK_i \nonumber \\
\Rightarrow & \epsilon_i\leq \frac{CK_i}{|\mathcal{L}_n|} \nonumber 
\end{align}
which maximizes $|\mathcal{L}_n|$. Consider
\begin{align}
&1+ \sum_{j\in \mathcal{L}_n,j\neq i}x_ja_{j,i}\nonumber \\
&= 1+\sum_{j\in \mathcal{L}_s,j\neq i}x_ja_{j,i}+\sum_{j\in \mathcal{L}_n\backslash\mathcal{L}_s,j\neq i}x_ja_{j,i} \nonumber \\
&\geq 1+\sum_{j\in \mathcal{L}_s,j\neq i}x_ja_{j,i}+(|\mathcal{L}_n\backslash\mathcal{L}_s|-1)\epsilon_i
\end{align}
Thus we have
\begin{align}
&1+\sum_{j\in \mathcal{L}_s,j\neq i}x_ja_{j,i}+(|\mathcal{L}_n\backslash\mathcal{L}_s|-1)\epsilon_i \nonumber \\
&\leq K_i + (|\mathcal{L}_n\backslash\mathcal{L}_s|-1)\epsilon_i \leq CK_i
\end{align}
Summation over all $i$ yields
\begin{align}
C\sum_iK_i &\geq \sum_i (K_i+(|\mathcal{L}_n\backslash\mathcal{L}_s|-1)\epsilon_i)\nonumber \\
&\geq \sum_i K_i+(|\mathcal{L}_n\backslash\mathcal{L}_s|-1)\sum_i\epsilon_i \nonumber \\
&\geq \sum_i K_i+(|\mathcal{L}_n\backslash\mathcal{L}_s|-1)\frac{C}{|\mathcal{L}_n|}\sum_iK_i \nonumber \\
\Rightarrow C&\geq 1 + (|\mathcal{L}_n\backslash\mathcal{L}_s|-1)\frac{C}{|\mathcal{L}_n|}
\end{align}
%\begin{align}
%&\sum_i (K_i+(|\mathcal{L}_n\backslash\mathcal{L}_s|-1)\epsilon_i)\leq C\sum_iK_i \nonumber \\
%&\sum_i K_i+(|\mathcal{L}_n\backslash\mathcal{L}_s|-1)\sum_i\epsilon_i\leq C\sum_iK_i \nonumber \\
%&\sum_i K_i+(|\mathcal{L}_n\backslash\mathcal{L}_s|-1)\frac{C}{|\mathcal{L}_n|}\sum_iK_i\leq C\sum_iK_i \nonumber \\
%&1 + (|\mathcal{L}_n\backslash\mathcal{L}_s|-1)\frac{C}{|\mathcal{L}_n|}\leq C
%\end{align}

Since $\mathcal{L}_s$ is a subset of $\mathcal{L}_n$, $|\mathcal{L}_n\backslash\mathcal{L}_s| = |\mathcal{L}_n|-|\mathcal{L}_s|$, and we get
\begin{equation}
|\mathcal{L}_n|\leq C(|\mathcal{L}_s|+1)
\end{equation}\qed

%\begin{align}
%&1+ \sum_{j\in \mathcal{L}_n,j\neq i}x_ja_{j,i} = 1+\sum_{j\in \mathcal{L}_s,j\neq i}x_ja_{j,i}\nonumber\\
%&+\sum_{j\in \mathcal{L}_n\backslash\mathcal{L}_s,j\neq i}x_ja_{j,i} = \bigl(1+\sum_{j\in \mathcal{L}_s,j\neq i}x_ja_{j,i}\bigr)\nonumber\\
%&+\bigl(1+\sum_{j\in \mathcal{L}_n\backslash\mathcal{L}_s,j\neq i}x_ja_{j,i}\bigr)-1
%\end{align}
%Thus we have
%\begin{align}
%&\bigl(1+\sum_{j\in \mathcal{L}_s,j\neq i}x_ja_{j,i}\bigr) + |\mathcal{L}_n\backslash\mathcal{L}_s|\epsilon-1\nonumber\\
%&\leq K+ |\mathcal{L}_n\backslash\mathcal{L}_s|\epsilon-1\leq CK
%\end{align}
%Since $\mathcal{L}_s$ is a subset of $\mathcal{L}_n$, $|\mathcal{L}_n\backslash\mathcal{L}_s| = |\mathcal{L}_n|-|\mathcal{L}_s|$ and we get
%\begin{equation}
%OPT\leq |\mathcal{L}_n|\leq \frac{|\mathcal{L}_s|CK}{K-1}
%\end{equation}
%\qed
%Then we have for each $i$ in $\mathcal{L}_n$
%\begin{align}
%&K + |\mathcal{L}_n\backslash\mathcal{L}_s|\epsilon\leq 1+\sum_{j\in \mathcal{L}_n,j\neq i}x_ja_{j,i}\leq CK
%\end{align}
%since $1+\sum_{j\in \mathcal{L}_s,j\neq i}x_ja_{j,i}\leq K$ and $|\mathcal{L}_n\backslash\mathcal{L}_s|\epsilon\leq 1+\sum_{j\in \mathcal{L}_n\backslash\mathcal{L}_s,j\neq i}x_ja_{j,i}$. 
%We then have that
%\begin{align}
%&|\mathcal{L}_n\backslash\mathcal{L}_s|\epsilon_s + K\leq CK \nonumber \\
%&|\mathcal{L}_n\backslash\mathcal{L}_s|\epsilon_s\leq K(C-1) \nonumber \\
%&(|\mathcal{L}_n|-|\mathcal{L}_s|)\frac{K}{|\mathcal{L}_s|}\leq K(C-1) \nonumber \\
%&\frac{|\mathcal{L}_n|}{|\mathcal{L}_s|}K-K\leq K(C-1) \nonumber \\
%&|\mathcal{L}_n|\leq C|\mathcal{L}_s| \nonumber 
%\end{align}